\documentclass[12pt]{article}

\usepackage[margin=1.15in]{geometry}
\usepackage{amsmath,amssymb,amsthm,mathtools}
\usepackage{bm}
\usepackage{enumitem}
\usepackage{booktabs,array}
\usepackage{caption,subcaption}
\usepackage[authoryear,round]{natbib}
\usepackage[colorlinks=true,linkcolor=blue!50!black,citecolor=blue!50!black,urlcolor=blue!50!black]{hyperref}
\usepackage{xcolor}

\setlist{itemsep=2pt,topsep=4pt}

\theoremstyle{plain}
\newtheorem{theorem}{Theorem}
\newtheorem{proposition}{Proposition}
\newtheorem{lemma}{Lemma}
\newtheorem{corollary}{Corollary}
\theoremstyle{definition}
\newtheorem{definition}{Definition}
\newtheorem{assumption}{Assumption}
\newtheorem{example}{Example}
\theoremstyle{remark}
\newtheorem{remark}{Remark}
\usepackage{comment}

\DeclareMathOperator{\supp}{supp}
\DeclareMathOperator*{\argmin}{arg\,min}
\DeclareMathOperator*{\argmax}{arg\,max}
\DeclareMathOperator{\relint}{relint}
\DeclareMathOperator{\spn}{span}

\newcommand{\R}{\mathbb{R}}
\newcommand{\N}{\mathbb{N}}

\newcommand{\E}{\mathbb{E}}
\newcommand{\one}{\mathbf{1}}
\newcommand{\D}{\mathcal{D}}
\newcommand{\F}{\mathcal{F}}

\newcommand{\CE}{\mathrm{CE}}
\newcommand{\JC}{\mathrm{JC}}
\newcommand{\IMPL}{\mathrm{IMPL}}
\newcommand{\DIR}{\mathrm{DIR}}
\newcommand{\BR}{\mathrm{BR}}
\newcommand{\KL}{\mathrm{KL}}
\newcommand{\ac}{\ll}
\newcommand{\marg}{\text{marg}}

\title{Coordination Mechanisms with\\ Partially Specified Probabilities}
\author{Francesco Giordano\thanks{HEC Paris, Economics and Decision Sciences. Email: \href{mailto:francesco.giordano@hec.edu}{francesco.giordano@hec.edu}. I am grateful to Frederic Koessler, Tristan Tomala, Marie Laclau, Nicolas Vieille, Itzhak Gilboa, Fran\c{c}oise Forges, Marco Scarsini, Julian Chitiva, and Francesco Conti for helpful discussions and comments. This work was significantly improved through discussions with several invited speakers at the HEC weekly seminar. I thank participants of the (Junior) Game Theory Seminar at the Institut Henri Poincar\'e, the HEC PhD brownbag, the 7th World Congress of the Game Theory Society (Beijing), the Premi\`ere rencontre nationale du RT Optimisation (Lyon), the 2025 European Winter Meeting of the Econometric Society (Nicosia), and the 2026 Asian Summer School of Economic Theory (Hong Kong). This work benefited from a State grant managed by the Agence Nationale de la Recherche under the Investissements d'Avenir programme, reference ANR-18-EURE-0005 / EUR DATA EFM. All errors are my own. }}
\date{\today\\[0.5em]}

\begin{document}
\maketitle

\begin{abstract}
\noindent
A mediator coordinates play in a finite game, but instead of revealing a fully known correlation device it discloses only some statistics of the process that generates players' private messages. Players complete the partially specified process by minimizing a loss over the set of distributions consistent with what is disclosed: maximum Shannon entropy is the leading case, and the analysis covers every strictly convex and symmetric loss that never excludes an event the data do not exclude. We obtain four sets of results. (i) With unrestricted message spaces, the implementable outcomes are exactly the \emph{jointly coherent} outcomes, namely distributions supported on the support of the set of correlated equilibria. (ii) The characterization is unchanged when players' partial knowledge also concerns the joint distribution of messages and actions, whether or not players know the strategy profile. (iii) Under direct implementation, implementability depends on the loss and reduces to a single optimality condition linking the outcome to some correlated equilibrium: among all distributions on the line joining the equilibrium and the outcome, the equilibrium must be the one with the lowest loss. (iv) In single-agent decision problems the designer can induce, with probability one, any action that is a best response to some interior belief.
\medskip


\end{abstract}

\newpage
\section{Introduction}\label{sec:intro}

Decision makers often act on data whose generating process they understand only partially. A trader receives an analyst's recommendation knowing that other traders receive recommendations built from overlapping data, but not how strongly they are correlated. Members of a committee read reports from consultants who share sources, without knowing the extent of the overlap. Users of a platform see posts selected by an opaque algorithm and know that others see related, but not identical, content. In each case the strategic environment is understood, but the joint distribution of the private data (i.e., the correlation device) is only partially specified, and players need some rule to fill in what they do not know.

This paper asks which outcomes can be implemented when a designer controls not only the data-generating process but also how much of it is disclosed. We consider an idealized information provider that draws private messages from a true joint distribution $\eta$ and discloses the expectations, under $\eta$, of a collection $\F$ of random variables (a \emph{feedback structure}). The disclosed moments define the set of distributions that players regard as plausible. Players then select a single belief from this set by minimizing a loss function, and best respond to the selected belief. Actual play, however, is generated by the true process. The wedge between the true process and the selected belief is the designer's instrument.

The leading case for belief selection is maximum Shannon entropy \citep{jaynes1968prior}: under the maximum entropy principle, when only the marginals are known, the selected belief is their product, capturing correlation neglect; when only the support is known, the selected belief is uniform, reflecting the principle of insufficient reason. We study a class of loss functions that are strictly convex, invariant under relabeling of message profiles, and steep at the boundary of the simplex, which ensures a grain of truth: the selected belief assigns zero probability to an event only when the disclosure rule forces it to have zero probability. By the classical correspondence between convex functions and proper scoring rules \citep{mccarthy1956measures,savage1971elicitation,gneiting2007strictly}, minimizing such a loss over a convex set of plausible distributions is equivalent to choosing the robust-Bayes forecast in a fictitious prediction problem \citep{grunwald2004game}. 

\paragraph{Results.} Our first result (Theorem~\ref{thm:main}) characterizes the implementable outcomes when message spaces are unrestricted. An outcome is implementable only if it is \emph{jointly coherent} in the sense of \citet{nau1990coherent}: every action profile it puts weight on lies in the support of some correlated equilibrium. Conversely, every jointly coherent outcome is approximately implementable, and exactly implementable when payoffs are rational. The two directions use different properties of the loss: necessity follows from steepness alone, through a grain-of-truth property, sufficiency from symmetry and strict convexity, and neither depends on the specific functional form. Hence the implementable set is independent on the specific loss in the class. The construction uses a single indicator statistic (i.e., players are told some message profiles that never occur) and pure strategies. 

Our second result (Proposition~\ref{prop:endogenous}) extends the analysis to the case in which partial knowledge concerns the joint distribution of messages and actions, an endogenous variable, as in the literature on misspecified equilibrium \citep{esponda2008behavioral,spiegler2021modeling}. Whether or not players know the strategy profile, the implementable set is again the set of jointly coherent outcomes.

Our third result (Proposition~\ref{prop:direct}) concerns direct implementation, in which messages are action recommendations that players obey. Here the revelation principle fails, and the answer depends on the loss: an outcome $\mu$ is directly implementable if and only if there is a correlated equilibrium $q$ with $\supp\mu\subseteq\supp q$ that minimizes the loss on the line of the simplex through $q$ and $\mu$. For differentiable losses this is a Bregman-divergence condition, and under Shannon entropy it is the cross-entropy condition $\E_\mu[\log q]=\E_q[\log q]$. 

Our fourth result (Proposition~\ref{prop:decision}) applies the framework to single-agent decision problems with a payoff-relevant state and a known prior, connecting it to Bayesian persuasion \citep{kamenica2011bayesian}. Whenever the prior and the payoffs have rational values, any action that is a best response to some interior belief can be induced with probability one, and any distribution over such actions can be implemented. 

\paragraph{Related literature.}
Correlated behavior and its mediation interpretation originate with \citet{aumann1974subjectivity,aumann1987correlated} and \citet{myerson1982optimal}; see \citet{forges2024subjectivity} for a recent account. \citet{nau1990coherent} show that correlated equilibrium is the dual of a no-arbitrage (joint coherence) condition on the play of a game, in the same sense in which a probability is the dual of coherent bets \citep{definetti1937prevision}; our main characterization gives joint coherence an implementation interpretation.

Partially specified probabilities were introduced by \citet{lehrer2012partially}, whose players evaluate actions by the worst case over distributions consistent with partial information. \citet{koessler2025feedback} design feedback for maxmin players and characterize maxmin self-confirming equilibria. Our players instead select one belief by a heuristic and maximize expected utility, an approach in the spirit of analogy-based expectations \citep{jehiel2005analogy} and of maximum-entropy completion of limited data access in \citet{spiegler2021modeling}. Relative to the latter, we take the feedback structure to be the designer's choice variable. The idea that a designer can control players' feedback appears in \citet{jehiel2011manipulative} and \citet{jehiel2024auction}. Self-confirming and misspecified equilibrium concepts \citep{fudenberg1993self,battigalli2015self,esponda2008behavioral} motivate the study of partially specified play (Section~\ref{sec:endogenous}).

The use of maximum entropy for prior selection goes back to \citet{jaynes1968prior,jaynes1982rationale}, with axiomatizations by \citet{shore1980axiomatic}, \citet{skilling1988axioms}, and \citet{csiszar1991least}; \citet{grunwald2004game} establish the duality between generalized-entropy maximization and robust Bayes decisions under proper scoring rules, which underlies our class of losses. Misspecified or coarse beliefs in persuasion are studied by \citet{eliaz2021persuasion,eliaz2021strategic} and \citet{levy2022persuasion}; correlation neglect is documented by \citet{enke2019correlation}, and ambiguity about correlation or interpretation is studied by \citet{epstein2019ambiguous,epstein2024hard}. Mechanism design with ambiguity appears in \citet{bose2014mechanism}, \citet{carroll2017robustness} and \citet{dutting2024ambiguous}.

\paragraph{Outline.} Section~\ref{sec:model} presents the model, the class of belief-formation rules, and two introductory examples. Section~\ref{sec:implementation} contains the main characterization and Section~\ref{sec:endogenous} studies partially specified play. Section~\ref{sec:direct} studies the case direct implementation, while Section~\ref{sec:decision} focuses on decision problems 
and Section~\ref{sec:conclusion} concludes. Proofs are in the Appendix.

\section{Model}\label{sec:model}

\subsection{The game}
Let $\Gamma=(N,A,u)$ be a finite game with players $N=\{1,\dots,n\}$, $n\ge 2$, action profiles $A=\times_{i\in N}A_i$, and payoffs $u_i:A\to\R$. Payoffs are extended multilinearly to mixed actions. An \emph{outcome} is a distribution $\mu\in\Delta A$. For a finite set $X$ and $q\in\Delta X$, $\supp q=\{x:q_x>0\}$; for $q,\eta\in\Delta X$, $\eta\ac q$ means $\supp\eta\subseteq\supp q$.

\subsection{Partially specified data-generating processes}
A designer chooses finite message sets $M_i$, $i\in N$, with $M=\times_i M_i$ the set of message \emph{profiles}, and a distribution $\eta\in\Delta M$. Rather than revealing $\eta$, the designer discloses the expectations $\E_\eta[f]$ of the random variables in a \emph{feedback structure} $\F\subseteq\{f:M\to\R\}$.

\begin{definition}[Partially specified data-generating process]
A \emph{partially specified data-generating process} (PSDGP) is a triple $\D=(M,\eta,\F)$. The set of \emph{plausible} distributions is
\[
\Delta_\D=\bigl\{q\in\Delta M:\ \E_q[f]=\E_\eta[f]\ \text{for all } f\in\F\bigr\}.
\]
\end{definition}

The set $\Delta_\D$ is a nonempty (it contains $\eta$) convex polytope. The feedback structure is public: all players know the same moments. The PSDGP is \emph{fully specified} when the moments identify $\eta$, that is, $\Delta_\D=\{\eta\}$. Full specification occurs, for example, if $\F=\{\one_{\{m\}}:m\in M\}$.

\paragraph{What feedback can express.} Several familiar forms of partial knowledge are special cases. Disclosing the probability of a collection of events $E_1,\dots,E_k\subseteq M$ corresponds to $\F=\{\one_{E_1},\dots,\one_{E_k}\}$; in particular, some specific partitions of $M$ give an analogy-class feedback in the sense of \citet{jehiel2005analogy}, while disclosing the expectations of the indicators of the message combinations received by a group of players gives the marginals over subsets of variables, as in \citet{spiegler2021modeling}. Disclosing that some message profiles never occur corresponds to $\one_Z$ with $\E_\eta[\one_Z]=0$. The numerical values of the functions in $\F$ matter only through the linear space they span, as the next lemma shows.

\begin{lemma}\label{lem:span}
Let $V_\F=\spn(\F\cup\{\one_M\})$. Then $\Delta_\D=\{q\in\Delta M:\E_q[f]=\E_\eta[f]\ \forall f\in V_\F\}$. In particular:
\begin{enumerate}[label=(\alph*)]
    \item two feedback structures with the same span yield the same plausible set for every $\eta$;
    \item there is a finite $\F'\subseteq\F$ with $|\F'|\le|M|$ such that $\Delta_{(M,\eta,\F')}=\Delta_{(M,\eta,\F)}$ for every $\eta\in\Delta M$.
\end{enumerate}
\end{lemma}

A feedback structure is thus finite without loss of generality. What the formulation leaves out are inequality restrictions and nonlinear statistics. 

\subsection{Belief formation}\label{sec:belief}

Players complete the partially specified process by selecting a single belief from $\Delta_\D$. Belief formation occurs ex ante, before messages are received. Upon receiving $m_i$, player $i$ updates the selected belief by Bayes' rule.

\begin{definition}[Belief-formation rule]
A \emph{loss} is a family $L=(L_M)$ indexed by finite sets of message profiles, with $L_M:\Delta M\to\R$. The belief selected under $\D=(M,\eta,\F)$ is\footnote{Under the regularity conditions that follows, such minimizer exists and is unique.}
\[
q^\D=\argmin_{q\in\Delta_\D} L_M(q).
\]
\end{definition}

We impose two sets of conditions. The first is a natural condition needed for the belief to be well defined and contain a "grain-of-truth" (i.e., not assign probability zero to events unless forced by the feedback structure); the second adds symmetry to relabeling of messages.
\begin{assumption}[Regular loss]
\label{ass:regular}
For every finite $M$, the following regularity conditions hold:
\begin{enumerate}
    \item $L_M$ is continuous and strictly convex on $\Delta M$;
    \item For any $S\subseteq M$, the restriction of $L_M$ to $\Delta S$ is differentiable on the relative interior of $\Delta S$;
    \item  $L_M$ is steep: for every $q\in\Delta M$ and all $m,k\in M$ with $q_m=0<q_k$, it holds\footnote{Where $e_j$ denotes the unit vector corresponding to on element $j$.} \[\lim_{t\to 0}\ \frac{L_M\bigl(q+t(e_m-e_k)\bigr)-L_M(q)}{t}=-\infty.\]

\end{enumerate}
\end{assumption}
Steepness is a condition on the derivative of $L_M$ with respect to the probability of a single message profile $m$, computed while remaining in the simplex: an infinitesimal amount of probability is moved to $m$, which the belief excludes, from a profile $k$ that it considers possible, and the loss falls at an infinite rate. Steepness yields the grain-of-truth property of Lemma~\ref{lem:properties}(a).

\begin{assumption}[Symmetry]\label{ass:symmetric}
For every finite $M$ and any permutation $\pi$ of $M$, $L_M(q\circ\pi)=L_M(q)$ for all $q\in\Delta M$.
\end{assumption}

Symmetry is with respect to all permutations of the set of message profiles. Since only the minimizer matters, any strictly increasing transformation of a loss defines the same belief formation rule.

\begin{remark}[Separable losses]\label{rem:separable}
Let $g:[0,1]\to\R$ be continuous, differentiable and strictly convex, with $\lim_{q_m\to0^+}g'(q_m) =-\infty$. Then $L_M(q)=\sum_{m\in M}g(q_m)$ satisfies Assumptions~\ref{ass:regular} and~\ref{ass:symmetric}, and so does every strictly increasing and convex transformation of $L$. The class of losses satisfying the two assumptions is larger, as it also contains symmetric losses that are not of this form.
\end{remark}

A notable example is the Shannon entropy, with $g(t)=t\log t$, so that $L_M=-H$ and $q^\D$ is the maximum-entropy belief. Other notions of entropy are also covered, such as the Tsallis entropy, with $g(t)=-t^\alpha$ for $\alpha\in(0,1)$, and the R\'enyi entropy of order $\alpha\in(0,1)$, which is a strictly increasing transformation of the Tsallis entropy.

The following lemma collects some properties of the belief formation rule used throughout.

\begin{lemma}\label{lem:properties}
Let $\D=(M,\eta,\F)$.
\begin{enumerate}[label=(\alph*)]
\item Under Assumption~\ref{ass:regular}, $q^\D$ exists and is unique, and it has maximal support in $\Delta_\D$\footnote{In the sense that $\supp(q^\D) = \supp(\Delta_\D)$.}. In particular, $\eta\ac q^\D$.
\item Under Assumptions~\ref{ass:regular} and~\ref{ass:symmetric}, let $\{E_0,E_1,\dots,E_k\}$ be disjoint subsets of $M$ and suppose $\Delta_\D=\{q\in\Delta M: q(E_j)=c_j,\ j=0,\dots,k\}$. Then $q^\D_m=c_j/|E_j|$ for every $m\in E_j$.
\end{enumerate}
\end{lemma}

Part (a) is a grain-of-truth property: no message profile that the data do not rule out receives zero probability. Part (b) is the principle of insufficient reason and holds for every symmetric and convex loss.

\paragraph{Inducible beliefs.} A key step in several results is to understand which beliefs a designer can induce from a given true process.

\begin{definition}
A pair $(q,\eta)\in\Delta M\times\Delta M$ is \emph{admissible} for $L$ if there is $\F$ such that $q=q^{(M,\eta,\F)}$.
\end{definition}

For $q\in\Delta M$ with support $S\subseteq M$, let $\frac{\partial L_M(q)}{\partial q_m}$, $m\in S$, denote the partial derivatives at $q$ of the restriction of $L_M$ to $\Delta S$. For $\eta\ac q$, let
\[
\varphi_L(q,\eta)=\sum_{m\in S}\frac{\partial L_M(q)}{\partial q_m}\,(q_m-\eta_m),\qquad D_L(\eta,q)=L_M(\eta)-L_M(q)+\varphi_L(q,\eta),
\]
where $D_L$ is the Bregman divergence induced by $L_M$. The following lemma studies pairs of beliefs and data-generating processes that are admissible for a given loss.

\begin{lemma}\label{lem:admissible}
Under Assumption~\ref{ass:regular}, the following are equivalent: 
\begin{enumerate}[label=(\alph*)]
    \item $(q,\eta)$ is admissible;
    \item $\eta\ac q$ and $\varphi_L(q,\eta)=0$;
    \item $\eta\ac q$ and $D_L(\eta,q)=L_M(\eta)-L_M(q)$.
\end{enumerate}
In particular, if $(q,\eta)$ is admissible, belief $q$ can be induced from $\eta$ by disclosing its support and the single statistic $f$ with $f_m=\frac{\partial L_M(q)}{\partial q_m}$ for $m\in S$ and $f_m=0$ otherwise.
\end{lemma}

The condition $\varphi_L(q,\eta)=0$ says that the derivative of the loss at $q$ in the direction of $\eta$ is zero. In other words, the belief from $q$ toward the true process, or away from it, does not lower the loss to a first order approximation.

\begin{corollary}[Shannon entropy]\label{cor:shannon-admissible}
Under Shannon entropy, $(q,\eta)$ is admissible if and only if $\eta\ac q$ and $\E_\eta[\log q]=\E_q[\log q]$; equivalently, the Kullback-Leibler divergence between $\eta$ and $q$ equals the difference in Shannon entropy, $\KL(\eta, q)=H(q)-H(\eta)$.
\end{corollary}

\paragraph{Discussion of the assumptions.} Strict convexity makes the rule single-valued and captures a preference for compromise: if two beliefs are plausible, their mixture is preferred to the corresponding mixture of their losses. Symmetry expresses indifference to labels: when the data do not distinguish two message profiles, the selected belief does not either. Steepness delivers the grain-of-truth property of Lemma~\ref{lem:properties}(a). By the McCarthy--Savage representation, every strictly convex $L$ is the negative of the generalized entropy of a strictly proper scoring rule, and minimizing $L$ over the convex set $\Delta_\D$ is equivalent to choosing the forecast that maximizes the worst-case expected score over $\Delta_\D$ \citep{grunwald2004game}. A player with such a rule behaves as a maxmin forecaster in a fictitious elicitation problem that treats message profiles symmetrically and penalizes zero-probability forecasts, and then carries the resulting belief into the actual game. 

A feature of every symmetric rule, emphasized by the principle of insufficient reason, is that it is not invariant to refinements of the message space: splitting a message into payoff-equivalent copies changes the selected belief. In our setting the designer chooses the message space, so the number of synonymous messages is part of the design. In fact, it is the channel through which the constructions of Section~\ref{sec:implementation} operate. Section~\ref{sec:direct} shows what is implementable by restricting messages to action recommendations.

\subsection{Two introductory examples}\label{sec:intro-examples}

\begin{example}[Chicken and insufficient reason]\label{ex:chicken}
Two firms decide whether to expand capacity ($a_2$, $b_2$) or stay small ($a_1$, $b_1$) in a market that can absorb one expansion. Payoffs are in Table~\ref{tab:chicken}(a); the welfare-maximizing correlated equilibrium is uniform on $\{(a_1,b_1),(a_1,b_2),(a_2,b_1)\}$. An industry analyst sends each firm a private recommendation $m_i\in\{\text{small},\text{expand}\}$. The true process always recommends ``small'' to both (Table~\ref{tab:chicken}(b)), and the only statistic disclosed is that both firms are never told to expand, $\F=\{\one_{(\text{expand},\text{expand})}\}$ with expectation zero. By Lemma~\ref{lem:properties}(b), every loss satisfying Assumptions~\ref{ass:regular}--\ref{ass:symmetric} selects the uniform belief on the remaining profiles (Table~\ref{tab:chicken}(c)). Under this belief obeying the recommendation is optimal, because the belief coincides with the optimal correlated equilibrium. The outcome $(a_1,b_1)$ is implemented, although it is not a correlated equilibrium; it gives each firm $5$, more than the $14/3$ each obtains in the welfare-maximizing correlated equilibrium.
\end{example}

\begin{table}[ht]\centering\small
\begin{tabular}{c@{\qquad}c@{\qquad}c}
\begin{tabular}{c|cc}
 & $b_1$ & $b_2$\\\hline
$a_1$ & $5,5$ & $2,7$\\
$a_2$ & $7,2$ & $0,0$
\end{tabular}
&
\begin{tabular}{c|cc}
 & $b_1$ & $b_2$\\\hline
$a_1$ & $1$ & $0$\\
$a_2$ & $0$ & $0$
\end{tabular}
&
\begin{tabular}{c|cc}
 & $b_1$ & $b_2$\\\hline
$a_1$ & $1/3$ & $1/3$\\
$a_2$ & $1/3$ & $0$
\end{tabular}\\[2pt]
(a) Payoffs & (b) True process $\eta$ & (c) Selected belief $q^\D$
\end{tabular}
\caption{Chicken. Messages are identified with recommendations.}\label{tab:chicken}
\end{table}

\begin{example}[Coordination and correlation neglect]\label{ex:coordination}
Two divisions choose a technical standard; the game is in Table~\ref{tab:coord}(a). Each division receives a recommendation from a consultant; consultants use overlapping data, and only the frequency with which each recommendation is issued is disclosed: $\F$ consists of the marginal indicators. The true process is perfectly anti-correlated (Table~\ref{tab:coord}(b)). Under Shannon entropy the selected belief is the product of the marginals \citep[Ch.~12]{cover2006elements}, which is the mixed Nash equilibrium of the game; following recommendations is therefore optimal, and the miscoordinated outcome $\eta$ is implemented. No correlated equilibrium puts all its weight on the anti-diagonal, but both anti-diagonal profiles lie in the support of the (full-support) mixed equilibrium. This example uses correlation neglect and is specific to Shannon entropy; the general results below rely only on the more general class of strictly convex, symmetric and steep loss functions.
\end{example}

\begin{table}[ht]\centering\small
\begin{tabular}{c@{\qquad}c@{\qquad}c}
\begin{tabular}{c|cc}
 & $b_1$ & $b_2$\\\hline
$a_1$ & $2,1$ & $0,0$\\
$a_2$ & $0,0$ & $1,2$
\end{tabular}
&
\begin{tabular}{c|cc}
 & $b_1$ & $b_2$\\\hline
$a_1$ & $0$ & $2/3$\\
$a_2$ & $1/3$ & $0$
\end{tabular}
&
\begin{tabular}{c|cc}
 & $b_1$ & $b_2$\\\hline
$a_1$ & $2/9$ & $4/9$\\
$a_2$ & $1/9$ & $2/9$
\end{tabular}\\[2pt]
(a) Payoffs & (b) True process $\eta$ & (c) Shannon belief $q^\D$
\end{tabular}
\caption{Coordination with disclosed marginals.}\label{tab:coord}
\end{table}

\section{Implementation}\label{sec:implementation}

A strategy of player $i$ is a map $\sigma_i:M_i\to\Delta A_i$; a profile $\sigma$ induces $\sigma(a\mid m)=\prod_i\sigma_i(a_i\mid m_i)$. For $\eta\in\Delta M$, the induced outcome is $(\eta\circ\sigma)(a)=\sum_m\eta_m\sigma(a\mid m)$.

\begin{definition}[Implementation]\label{def:impl}
An outcome $\mu\in\Delta A$ is \emph{implemented} by $\D=(M,\eta,\F)$ if there is a strategy profile $\sigma$ such that
\begin{enumerate}[label=(\roman*)]
\item \emph{belief formation:} $q=q^\D$;
\item \emph{best response:} for every $i\in N$, every $m_i$ with $q(m_i)>0$, and every $a_i'\in A_i$,
\[
\sum_{m_{-i}}q(m_i,m_{-i})\Bigl[u_i\bigl(\sigma_i(m_i),\sigma_{-i}(m_{-i})\bigr)-u_i\bigl(a_i',\sigma_{-i}(m_{-i})\bigr)\Bigr]\ \ge\ 0;
\]
\item \emph{outcome:} $\mu=\eta\circ\sigma$.
\end{enumerate}
The outcome is \emph{$\varepsilon$-implemented} if (ii) holds with $-\varepsilon q(m_i)$ on the right-hand side. $\IMPL$ and $\IMPL_\varepsilon$ denote the sets of outcomes implemented (resp.\ $\varepsilon$-implemented) by some PSDGP.
\end{definition}

The best-response condition is evaluated under the selected belief, while the outcome is generated by the true process. Players' beliefs are correct about the strategy profile (an assumption that is relaxed in the following section) and about the game; they are partially specified only about the (payoff-irrelevant) messages.

\begin{definition}[Correlated equilibrium; joint coherence]
An outcome $p\in\Delta A$ is a \emph{correlated equilibrium} if for all $i$, all $a_i$ with $p(a_i)>0$ and all $a_i'$, it holds $\sum_{a_{-i}}p(a_i,a_{-i})[u_i(a_i,a_{-i})-u_i(a_i',a_{-i})]\ge0$. Let $\CE$ be the set of correlated equilibria and $\supp\CE=\bigcup_{p\in\CE}\supp p$. An action profile is \emph{jointly coherent} if it belongs to $\supp\CE$ \citep{nau1990coherent}, and an outcome $\mu$ is jointly coherent if $\supp\mu\subseteq\supp\CE$. 
\end{definition}

Let $\JC$ denotes the set of jointly coherent outcomes. Notice that $\JC=\Delta(\supp\CE)$ and, in particular, it contains $\CE$ and can be larger. A fully specified process reduces the model to the standard one, so $\CE\subseteq\IMPL$. The following theorem characterizes the set of implementable outcomes: in any finite game, every implementable outcome is jointly coherent, and every jointly coherent outcome can be approximately implemented at any arbitrary approximation level. Furthermore, if the game has payoffs that are rational numbers, the set of implementable outcomes coincides with the set of jointly coherent outcomes.

\begin{theorem}\label{thm:main}
Under Assumptions~\ref{ass:regular} and~\ref{ass:symmetric}:
\begin{enumerate}[label=(\alph*)]
\item $\IMPL\subseteq\JC\subseteq\bigcap_{\varepsilon>0}\IMPL_\varepsilon$;
\item if $\Gamma$ has rational payoffs\footnote{That is, the utility of every player at any profile is a rational number.}, $\IMPL=\JC$.
\end{enumerate}
\end{theorem}
In particular, the implementable set does not depend on the specific functional form of the loss. Furthermore, an immediate consequence is that if a game has rational payoff, and the mediator has some payoff $v:A \to \R$, then he can attain $\max\{v(a): a \in \supp \CE\}$. 

The necessity part uses only Assumption~\ref{ass:regular}. If $(\D,\sigma)$ implements $\mu$, the best-response condition says that $\sigma$ is individually optimal for the belief $q^\D$, so the induced distribution $p=q^\D\circ\sigma$ is a correlated equilibrium. By steepness, $\eta\ac q^\D$, hence $\supp\mu=\supp(\eta\circ\sigma)\subseteq\supp p\subseteq\supp\CE$. The true process can differ from the belief, but it cannot put weight on profiles that the belief excludes. 

Sufficiency is constructive and is a consequence on the following lemma.

\begin{lemma}[Construction of a PSDGP]\label{lem:construction}
Under Assumptions~\ref{ass:regular} and~\ref{ass:symmetric}, let $p\in\CE$ have rational coordinates, so that, for each $a \in A$, $p_a=k_a/k$ with $k_a,k\in\N$. For each player $i$ and action $a_i$ let $M_i^{a_i}$ be a set of $k$ messages, $M_i=\bigcup_{a_i}M_i^{a_i}$, and $\sigma_i(m_i)=a_i$ for $m_i\in M_i^{a_i}$. Each block $M^a=\times_iM_i^{a_i}$ is identified with $\{1,\dots,k\}^n$. For each $a$, let $C^a\in\{0,1\}^{k^n}$ be a binary hypercube whose entries sum to $k_a$ along each line (that exists and can be constructed as in Lemma~\ref{lem:hypercube}), and let $Z=\{m\in M^a:\ C^a(m)=0,\ a\in A\}$. Then for any $\mu\in\Delta A$ with $\supp\mu\subseteq\supp p$ there is $\eta\in\Delta M$ such that $(M,\eta,\{\one_Z\})$ and the (pure) profile $\sigma$ implement $\mu$.
\end{lemma}

The message set attaches to each action profile $a$ a block $M^a$ of $k^n$ message profiles, and the zero set $Z$ removes from each block the profiles corresponding to the zeros of an auxiliary binary hypercube in which every line (all coordinates fixed but one) contains exactly $k_a$ ones. By insufficient reason, the selected belief is uniform on $M\setminus Z$, so each block receives probability $p_a$. Because line counts are proportional to $k_a$, each player's conditional belief about the others' recommended actions equals $p(\cdot\mid a_i)$. Obedience then follows from $p\in\CE$. The true process can put its mass anywhere in $M\setminus Z$, in particular on one profile per block with weights $\mu$. For two players the hypercubes are binary matrices with constant row and column sums \citep[Ch.~6]{ryser1963combinatorial}; for $n$ players they are built by induction on the dimension (Lemma~\ref{lem:hypercube}).

With rational payoffs the vertices of the polytope $\CE$ are rational, and their average is a rational correlated equilibrium with support $\supp\CE$; Lemma~\ref{lem:construction} then implements every jointly coherent outcome. Without rational payoffs, a rational distribution with the same support as a maximal-support correlated equilibrium is an approximate correlated equilibrium, which yields $\varepsilon$-implementation.

\begin{corollary}\label{cor:simple}
If $\Gamma$ has rational payoffs, every implementable outcome is implemented by a PSDGP in which strategies are pure, $|M_i|=k|A_i|$ for a common integer $k$, and the feedback structure is a single indicator (players are only told which message profiles never occur).
\end{corollary}
 
\begin{remark}[Role of each property]
Steepness drives necessity: if beliefs could exclude message profiles that occur, players would face no incentive constraint at the messages they actually receive, and outcomes outside $\JC$ could arise. Strict convexity makes the selected belief unique, and along with symmetry drives sufficiency through insufficient reason (Lemma~\ref{lem:properties}(b)). 
\end{remark}

The following example illustrates a construction based on feedback consisting of indicator functions.

\begin{example}[A construction based on indicators]\label{ex:indicator}
In the game of Table~\ref{tab:ind}(a), the unique correlated equilibrium is $q^*=(\tfrac16,\tfrac13,\tfrac16,\tfrac13)$. Consider the target $\mu$ uniform on $\{a_1b_1,a_1b_2,a_2b_1\}$. Each player has four messages; player 1 plays $a_1$ after $m_1,m_2$ and $a_2$ after $m_3,m_4$, and player 2 plays $b_1$ after $m_1',m_2'$ and $b_2$ after $m_3',m_4'$. The designer discloses that the four profiles $(m_1,m_2'),(m_2,m_1'),(m_3,m_2'),(m_4,m_1')$ never occur. By insufficient reason, the selected belief is uniform ($1/12$) on the remaining twelve profiles: it induces $q^*$ on actions with the conditional beliefs of $q^*$ at every message, so all strategies are best responses. The true process puts $1/3$ on each of $(m_1,m_1')$, $(m_1,m_3')$ and $(m_3,m_1')$, which implements $\mu$.
\end{example}

\begin{table}[ht]\centering\small
\begin{tabular}{c@{\qquad\qquad}c}
\begin{tabular}{c|cc}
 & $b_1$ & $b_2$\\\hline
$a_1$ & $2,0$ & $0,1$\\
$a_2$ & $0,1$ & $1,0$
\end{tabular}
&
\begin{tabular}{c|cccc}
 & $m_1'$ & $m_2'$ & $m_3'$ & $m_4'$\\\hline
$m_1$ & $1/3$ & $\times$ & $1/3$ & $0$\\
$m_2$ & $\times$ & $0$ & $0$ & $0$\\
$m_3$ & $1/3$ & $\times$ & $0$ & $0$\\
$m_4$ & $\times$ & $0$ & $0$ & $0$
\end{tabular}\\[2pt]
(a) Payoffs & (b) True process ($\times$: disclosed zeros)
\end{tabular}
\caption{A construction based on indicator functions.}\label{tab:ind}
\end{table}

\section{Partially specified play}\label{sec:endogenous}

So far players' partial knowledge concerns only the payoff-irrelevant messages. In much of the literature on misspecified beliefs, players instead hold coarse views of the joint distribution of exogenous and \emph{endogenous} variables \citep{jehiel2005analogy,eyster2005cursed, esponda2008behavioral,spiegler2021modeling}. We now let the disclosed statistics, and hence the belief, concern the joint distribution of messages and actions.

\begin{definition}[Extended PSDGP]
An \emph{extended PSDGP} is $\bar\D=(M,\bar\eta,\bar\F)$ with $\bar\eta\in\Delta(M\times A)$ and $\bar\F\subseteq\{f:M\times A\to\R\}$. 
\end{definition}
Player $i$ observes her data $(m_i,a_i)$ and takes an action. In equilibrium, the action component of the data is the action that is taken, so $\bar\eta$ is the joint distribution of data and play, and the disclosed statistics may concern both. The set of plausible distributions is $\Delta_{\bar\D}=\{q\in\Delta(M\times A):\E_q[f]=\E_{\bar\eta}[f]\ \text{for all }f\in\bar\F\}$, and the selected belief is $\bar q=\argmin_{q\in\Delta_{\bar\D}}L_{M\times A}(q)$. The selected belief $\bar q$ may put weight on pairs $(m,a)$ that never occur (players' data need not reveal how play relates to messages).

\begin{definition}[Implementation with partially specified play]\label{def:endo}
An outcome $\mu\in\Delta A$ is \emph{implemented} by $\bar\D$ if $\mu=\marg_A\bar\eta$ and, for all $i$, all $(m_i,a_i)$ with $\bar q(m_i,a_i)>0$, and all $a_i'$,
\[
\sum_{m_{-i},a_{-i}}\bar q(m,a)\bigl[u_i(a_i,a_{-i})-u_i(a_i',a_{-i})\bigr]\ \ge\ 0.
\]
Let $\overline{\IMPL}$ denote the set of outcomes implementable in this sense, and $\overline{\IMPL}_\varepsilon$ the set obtained when the best-response condition holds up to $-\varepsilon \bar q(m_i,a_i)$.
\end{definition}

The best-response condition requires every action that the belief deems possible to be optimal given the belief conditional on the player's data. If $\bar\F$ discloses that play is a function of messages, $a=\sigma(m)$ for a pure profile $\sigma$, the selected belief is supported on the graph of $\sigma$ and the model reduces to that of Section~\ref{sec:implementation}.

\begin{proposition}\label{prop:endogenous}
Under Assumptions~\ref{ass:regular} and~\ref{ass:symmetric}, if $\Gamma$ has rational payoffs, then $\overline{\IMPL}=\IMPL$.
\end{proposition}

With rational payoffs, restricting attention to pure strategies is without loss of generality (Corollary~\ref{cor:simple}).
Hence, in that case, partial specification of play neither enlarges nor shrinks the implementable set. 
On the other hand, $\overline{\IMPL}\subseteq\IMPL$ holds in every game, since data on play can always be relabeled as exogenous messages.
Furthermore, if payoffs are not rational, the construction applied to the approximate implementations of Theorem~\ref{thm:main}, which also use pure strategies and a single indicator, shows that every jointly coherent outcome can be approximated arbitrarily well with data on play. 

\section{Direct implementation}\label{sec:direct}

Under full specification the revelation principle holds: every correlated equilibrium is implemented by a canonical device that recommends actions and ensures that players are obedient. With partial specification, Examples ~\ref{ex:chicken} and ~\ref{ex:coordination} already show that outcomes outside $\CE$ can be implemented directly (i.e., through obedient action recommendations), while Theorem~\ref{thm:main} may require large message spaces. This section characterizes what direct implementation can achieve. Since there are no exogenous messages that may convey labels beyond the recommended action, direct implementation allow us to relax the symmetry assumption introduced in Section~\ref{sec:belief}. As symmetry plays no role in the analysis here, we work only under Assumption~\ref{ass:regular}.

\begin{definition}[Direct implementation]
An outcome $\mu\in\Delta A$ is \emph{directly implementable} if it is implemented by some $\D=(A,\eta,\F)$ together with the obedient profile $\sigma_i(a_i)=a_i$. Let $\DIR_L$ denote the set of directly implementable outcomes.
\end{definition}

With obedient strategies the outcome equals the true process, $\mu=\eta$, and the best-response condition says exactly that the selected belief $q^\D$ is a correlated equilibrium. Direct implementability therefore asks whether a correlated-equilibrium belief can be induced from $\mu$. For $q\in\Delta A$, let
\[
O_L(q)=\{\mu\in\Delta A:\ \mu\ac q\ \text{and}\ \varphi_L(q,\mu)=0\}
\]
be the set of true processes from which $q$ can be induced (by Lemma~\ref{lem:admissible}). For fixed $q$ with support $S$, the condition $\varphi_L(q,\mu)=\sum_{a\in S}\frac{\partial L(q)}{\partial q_a}(q_a-\mu_a)=0$ is linear in $\mu$. Hence $O_L(q)$ is the intersection of the face $\Delta S$ with a hyperplane through $q$.

\begin{proposition}\label{prop:direct}
Under Assumption~\ref{ass:regular}, the following are equivalent:
\begin{enumerate}[label=(\roman*)]
\item $\mu$ is directly implementable;
\item $\mu\in\bigcup_{q\in\CE}O_L(q)$\footnote{Consequently, by Lemma \ref{lem:admissible} this condition holds if and only if there exists $q\in\CE$ with $\mu\ac q$ and $D_L(\mu,q)=L(\mu)-L(q)$.}.
\end{enumerate}
\end{proposition}

Unlike Theorem~\ref{thm:main}, the set of outcomes implementable directly depends on the loss, and one has $\DIR_L\subseteq\IMPL$.

\begin{corollary}[Shannon entropy]\label{cor:direct-shannon}
Let $L$ be the negative Shannon entropy, then $\mu$ is directly implementable if and only if there is $q\in\CE$ with $\mu\ac q$ and $\E_\mu[\log q]=\E_q[\log q]$. Moreover, if $\Gamma$ has a unique correlated equilibrium $q^*$, $\DIR_H$ is the convex polytope $\{\mu\in\Delta A:\mu\ac q^*,\ \E_\mu[\log q^*]=\E_{q^*}[\log q^*]\}$.
\end{corollary}

\begin{example}[Direct implementation beyond correlated equilibrium]\label{ex:hml}
Two partners choose high ($h$), medium ($m$) or low ($l$) effort in the game of Table~\ref{tab:hml}(a). The true process recommends a common effort level, $\mu=(\tfrac14,0,0,0,\tfrac12,0,0,0,\tfrac14)$ on $(hh,hm,hl,mh,mm,ml,lh,lm,ll)$. This outcome is not a correlated equilibrium: a partner told $h$ who knew that the other was told $h$ would deviate to $l$ ($10>9$). The designer discloses the probabilities of three events: $E_1=\{hh,hm,mh\}$ (``at least one partner is told $h$ and nobody is told $l$''), $E_2=\{mm\}$, and $E_3$, the remaining profiles. Under $\mu$ these are $\tfrac14,\tfrac12,\tfrac14$. By insufficient reason (Lemma~\ref{lem:properties}(b)), every symmetric loss selects the belief in Table~\ref{tab:hml}(b), uniform within each event. This belief is a correlated equilibrium: a partner told $h$ is indifferent between $h$ and $l$ and prefers both to $m$, and a partner told $m$ or $l$ strictly prefers to follow. Hence $\mu$ is directly implementable under every loss satisfying Assumptions~\ref{ass:regular}--\ref{ass:symmetric}. Under Shannon entropy the same belief is induced by the single statistic $f=-\log_2 q$, which is constant on each event, as in Lemma~\ref{lem:admissible}; by Lemma~\ref{lem:span}, only the span of $\{f,\one\}$ matters, not the numerical values of $f$.
\end{example}

\begin{table}[ht]\centering\small
\begin{tabular}{c@{\qquad\qquad}c}
\begin{tabular}{c|ccc}
 & $h$ & $m$ & $l$\\\hline
$h$ & $9,9$ & $4,6$ & $1,10$\\
$m$ & $6,4$ & $6,6$ & $0,0$\\
$l$ & $10,1$ & $0,0$ & $6,6$
\end{tabular}
&
\begin{tabular}{c|ccc}
 & $h$ & $m$ & $l$\\\hline
$h$ & $1/12$ & $1/12$ & $1/20$\\
$m$ & $1/12$ & $1/2$ & $1/20$\\
$l$ & $1/20$ & $1/20$ & $1/20$
\end{tabular}\\[2pt]
(a) Payoffs & (b) Selected belief
\end{tabular}
\caption{Direct implementation beyond correlated equilibrium.}\label{tab:hml}
\end{table}

The following example illustrates that direct implementation depends on the loss function, and that additional outcomes can arise with leeway on the message set. 

\begin{table}[ht]\centering\small
\begin{tabular}{c|cc}
 & $b_1$ & $b_2$\\\hline
$a_1$ & $2,-2$ & $0,0$\\
$a_2$ & $0,0$ & $1,-1$
\end{tabular}
\caption{A matching-pennies game.}\label{tab:mp}
\end{table}
\begin{example}[Matching pennies]\label{ex:mp}
In the zero-sum game of Table~\ref{tab:mp}, the unique correlated equilibrium is the distribution $q^*=(1/9, 2/9, 2/9, 4/9)$ over the action profiles $(a_1b_1, a_1b_2, a_2b_1, a_2b_2)$. Under the negative Shannon entropy $H(q)=\sum_a q_a\log q_a$, Corollary~\ref{cor:direct-shannon} yields
\[
\DIR_H=\bigl\{\mu\in\Delta A:\ 2\mu_{a_1b_1}+\mu_{a_1b_2}+\mu_{a_2b_1}=\tfrac23\bigr\},
\]
which is a two-dimensional slice of the simplex passing through $q^*$. 

The set of directly implementable outcomes depends on the loss function. As an illustrative example, under the (negative) Tsallis entropy $T_\alpha(q)=-\sum_a q_a^{\alpha}$ with $\alpha \in (0,1)$, Lemma~\ref{lem:admissible} gives
\[
\DIR_{T_\alpha}=\Bigl\{\mu\in\Delta A:\ 9\mu_{a_1b_1}+2^{\alpha-1}\bigl(9\mu_{a_1b_2}+9\mu_{a_2b_1}-4\bigr)+4^{\alpha-1}\bigl(9\mu_{a_2b_2}-4\bigr)=1\Bigr\},
\]
representing another slice through $q^*$. For instance, the outcome $\mu=(0, 1/3, 1/3, 1/3)$ is directly implementable under negative Shannon entropy, but not under Tsallis entropy. The deterministic outcome $\delta_{a_1b_1}$ is contained in neither set. However, because this outcome is jointly coherent, it is implementable with a larger message space. 

To see this construction explicitly for the game of Table~\ref{tab:mp}, suppose each player has three messages: player 1 plays $a_1$ after $m_1$ and $a_2$ after $m_2, m_3$, while player 2 plays $b_1$ after $m_1'$ and $b_2$ after $m_2', m_3'$. With $\F=\emptyset$, the selected belief is uniform over the nine message profiles under any symmetric loss. Consequently, each player believes the opponent plays their first action with probability $1/3$, rendering both players indifferent and inducing the unique correlated equilibrium $q^*$. A true process that always draws $(m_1, m_1')$ then implements $\delta_{a_1b_1}$. 
\end{example}

\section{Decision problems under uncertainty}\label{sec:decision}

We now apply the framework to a single agent facing a payoff-relevant state, the setting of Bayesian persuasion \citep{kamenica2011bayesian}. Let $\Theta$ be a finite set of states with a known prior $p_0\in\Delta\Theta$ of full support, $A$ a finite set of actions and $u:A\times\Theta\to\R$. A designer chooses a finite message set $M$, a joint distribution $\eta\in\Delta(\Theta\times M)$ with $\mathrm{marg}_\Theta\eta=p_0$, and a feedback structure $\F\subseteq\{f:\Theta\times M\to\R\}$ that discloses the prior, in the sense that $\mathrm{marg}_\Theta q=p_0$ for every $q\in\Delta_\D$. The prior is known, while the likelihood of messages given states is partially specified. Let $\BR(p)=\argmax_{a}\E_p[u(a,\theta)]$.

\begin{definition}[Implementation in decision problems]
An outcome $\mu\in\Delta A$ is \emph{implemented} by $\D=(M,\eta,\F)$ if there is $\sigma:M\to\Delta A$ such that, with $q=q^\D$:
\begin{enumerate}[label=(\roman*)]
\item $\supp\sigma(m)\subseteq\BR\bigl(q(\cdot\mid m)\bigr)$ for every $m$ with $q(m)>0$;
\item $\mu(a)=\sum_{\theta,m}\eta(\theta,m)\,\sigma(a\mid m)$.
\end{enumerate}
\end{definition}

\begin{definition}
An action $a$ is \emph{rationalizable} if $a\in\BR(p)$ for some $p\in\Delta\Theta$, and \emph{interior-rationalizable} if $a\in\BR(p)$ for some $p$ of full support. 
\end{definition}
Let $R$ and $R^\circ\subseteq R$ denote, respectively, the sets of rationalizable and interior-rationalizable actions.
In a finite normal-form game, an action is a best response to a full-support belief if and only if it is not weakly dominated by any (possibly mixed) action. Hence $R^\circ=R$ whenever no rationalizable action is weakly dominated.

\begin{proposition}\label{prop:decision}
Under Assumptions~\ref{ass:regular} and~\ref{ass:symmetric}:
\begin{enumerate}[label=(\alph*)]
\item if $\mu$ is implementable, then $\supp\mu\subseteq R$;
\item if $u$ and $p_0$ have rational coordinates, every $\mu\in\Delta A$ with $\supp\mu\subseteq R^\circ$ is implementable;
\item if $u$ and $p_0$ have rational coordinates and $R=R^\circ$, the implementable outcomes are exactly $\Delta R$.
\end{enumerate}
\end{proposition}

In particular, the designer can induce any interior-rationalizable action with probability one, whatever its consequences in each state. Compare with Bayesian persuasion: there, the distribution of posteriors must average to the prior, that imposes a restriction on actions that are optimal only at beliefs far from $p_0$. Here the agent's belief about the state given a message is not disciplined by the true likelihood, only by the disclosed statistics, and Bayes plausibility imposes no restriction on the conditional beliefs that matter for behavior. The result parallels the full-manipulation result of \citet{levy2022persuasion} for receivers who neglect correlation among multiple signals, and it holds for every loss in our class. As in the ambiguous persuasion of \citet{beauchene2019ambiguous}, where the sender can choose an ambiguous experiment, here the sender can do better, for any loss satisfying Assumptions~\ref{ass:regular} and~\ref{ass:symmetric}.

The construction (Appendix~\ref{app:decision}) assigns to each action $a_b$ in the support of $\mu$ a block of messages and selects an interior belief $p^*_b$ with rational coordinates under which $a_b$ is optimal. The feedback discloses which (state, message) pairs never occur and the probability of each block-state cell $\{\theta\}\times M_b$. By insufficient reason the selected belief is uniform within each cell. Choosing the number of non-zero entries in each cell proportional to $p_0(\theta)/p^*_b(\theta)$ makes the conditional belief at the first message of block $b$ equal to $p^*_b$, while the true process sends all of block $b$'s mass, distributed across states according to the prior, to that message.

\begin{example}[Prosecutor and judge]\label{ex:judge}
A defendant is innocent ($\theta_I$) or guilty ($\theta_G$), with prior $p_0=(\frac23,\frac13)$. A judge chooses to convict ($c$) or acquit ($x$), with $u(c,\theta_G)=u(x,\theta_I)=1$ and $u(c,\theta_I)=u(x,\theta_G)=0$, so that she convicts if and only if the probability of guilt is at least $\frac12$. In Bayesian persuasion, the prosecutor convicts with probability at most $\frac23$ \citep{kamenica2011bayesian}.

Here, conviction can be induced with probability one. Let $M=\{m_1,m_2\}$ and let the true process send $m_1$ in both states (i.e., $\eta(\theta_I,m_1)=\frac23$ and $\eta(\theta_G,m_1)=\frac13$). The feedback discloses the prior and that a guilty defendant never receives $m_2$, that is, the probabilities of the cells $\{\theta_I\}\times M$, $\{(\theta_G,m_1)\}$ and $\{(\theta_G,m_2)\}$, equal to $\frac23$, $\frac13$ and $0$. It does not disclose how often an innocent defendant receives each message. By Lemma~\ref{lem:properties}(b), for every loss satisfying Assumptions~\ref{ass:regular} and~\ref{ass:symmetric}, the selected belief is uniform on each cell:
\[
q(\theta_I,m_1)=q(\theta_I,m_2)=q(\theta_G,m_1)=\tfrac13,\qquad q(\theta_G,m_2)=0.
\]
Table~\ref{tab:judge} summarizes the construction.
\begin{table}[ht]
\centering
\begin{tabular}{c|cc}
$\eta$ & $m_1$ & $m_2$\\ \hline
$\theta_I$ & $\frac23$ & $0$\\
$\theta_G$ & $\frac13$ & $0$
\end{tabular}
\qquad\qquad
\begin{tabular}{c|cc|c}
& $m_1$ & $m_2$ & $p_0$\\ \hline
$\theta_I$ & ? & ? & $\frac23$\\
$\theta_G$ & $\frac13$ & $0$ & $\frac13$
\end{tabular}
\qquad\qquad
\begin{tabular}{c|cc}
$q$ & $m_1$ & $m_2$\\ \hline
$\theta_I$ & $\frac13$ & $\frac13$\\
$\theta_G$ & $\frac13$ & $0$
\end{tabular}
\caption{True process (left), disclosed probabilities (center) and endogenous belief (right) in Example~\ref{ex:judge}. In the center panel, the marginals (i.e., the prior $p_0$) and the total probability of the cell $\{\theta_I\}\times M$ are disclosed.}
\label{tab:judge}
\end{table}

At $m_1$, the judge's belief is $q(\theta_G\mid m_1)=\frac12$, so conviction is optimal. At $m_2$ she believes the defendant innocent and acquits, but $m_2$ is never sent. Hence the judge convicts with probability one.
\end{example}

\begin{remark}[Weakly dominated actions]
Part (a) cannot be strengthened to $\Delta R$ in general. For example, consider the case where $\Theta=\{\theta_1,\theta_2\}$, $p_0=(\frac12,\frac12)$, $u(a,\theta_1)=1$, $u(a,\theta_2)=0$ and $u(b,\cdot)=1$. Action $a$ is rationalizable (by $\delta_{\theta_1}$) but weakly dominated. Implementing $\delta_a$ would require a message $m$ with $\eta(\theta_2,m)>0$, hence $q(\theta_2,m)>0$ by Lemma~\ref{lem:properties}(a), at which $a$ is not optimal.
\end{remark}

\section{Concluding remarks}\label{sec:conclusion}

When players complete partially specified data by a heuristic rule, disclosure policy becomes an instrument of coordination. For a broad class of rules (every strictly convex, symmetric and steep loss), this instrument has a sharp characterization that is independent on the specific functional form of the loss: the jointly coherent outcomes as introduced by \citep{nau1990coherent}. The limit is reached with simple feedback, it does not change when partial knowledge also concerns play, and its decision-theoretic counterpart replaces Bayes plausibility with a support condition. Direct implementation, by contrast, depends on the loss and is characterized by a single optimality condition (that, when the loss is differentiable, is a condition on the Bregman-divergence).

Several extensions remain open. The most natural is to Bayesian games, combining the payoff-relevant uncertainty of Section~\ref{sec:decision} with the strategic interaction of Sections~\ref{sec:implementation} and~\ref{sec:endogenous}, in the direction of Bayes correlated equilibrium \citep{bergemann2016bayes}. Dynamic settings, in which data accumulate from past play, would require reconciling the selected beliefs with what is learned over time. Finally, the designer's optimal choice among implementable outcomes, and the welfare consequences of coarse disclosure, deserve a separate treatment.

\appendix
\section{Proofs}\label{app:proofs}

\subsection{Proofs for Section~\ref{sec:model}}

\begin{proof}[Proof of Lemma~\ref{lem:span}]
(a) For $q\in\Delta M$, the map $f\mapsto\E_q[f]-\E_\eta[f]$ is linear and equal to zero at $f =\one_M$. By linearity, such map is equal to zero on $\F$ if and only if it is equal to zero on $V_\F$, that proves the first point. (b) For the second point, notice that $V_\F\subseteq\R^M$ has dimension at most $|M|$: choose a basis of $\spn\F$ consisting of elements of $\F$ and let $\F'$ be that basis. Then $\spn(\F'\cup\{\one_M\})=V_\F$, and the claim follows from the first point, for every $\eta$.
\end{proof}

\begin{proof}[Proof of Lemma~\ref{lem:properties}]
(a) $\Delta_\D$ is nonempty (as it contains $\eta$), convex and compact, and $L_M$ is continuous and strictly convex, so the minimizer exists and is unique. Let $q=q^\D$ and suppose that some $q'\in\Delta_\D$ has $q'_k>0=q_k$. Let $m$ be such that $q_m>0$ and let $v=q'_k(e_k-e_m)$. For $t\in(0,1)$, consider the mixture $(1-t)q+tq'$, add and subtract the term $t(1-t)v$, and obtain
\[
(1-t)q+tq'=(1-t)(q+tv)+t\,(q'-(1-t)v).
\] By strict convexity, one has
\[
L_M((1-t)q+tq')<(1-t)L_M(q+tv) + tL_M(q'+(1-t)v)),
\]
noticing that for $t$ sufficiently small, $q+tv,q'+(1-t)v \in \Delta(M)$, so the loss is well defined. Now, we can subtract on both sides the term $L_M(q)$ and divide by $t$ to obtain
\[
\frac{L_M((1-t)q+tq')-L_M(q)}{t}<(1-t)\frac{L_M(q+tv)-L_M(q)}{t}+L_M(q'-(1-t)v)-L_M(q).
\]
Notice that the first term on the RHS tends to $-\infty$ by steepness, while the term $L_M(q'-(1-t)v) - L_M(q)$ is bounded above since $L_M$ is continuous on the compact set $\Delta M$. Hence, for some $t$ small, one has $L_M((1-t)q+tq') - L_M(q)<0$ , while $(1-t)q+tq'\in\Delta_\D$ as $\Delta_\D$ is convex, that yields a contradiction.

(b) Let $\Pi$ be the set of permutations of $M$ that map each $E_j$ onto itself. Since $\Delta_\D$ is defined by the probabilities $q(E_j)$, every $\pi\in\Pi$ maps $\Delta_\D$ onto itself, and $L_M(q\circ\pi)=L_M(q)$ by Assumption~\ref{ass:symmetric}. Hence $q^\D\circ\pi$ is also a minimizer, and by uniqueness $q^\D\circ\pi=q^\D$ for all $\pi\in\Pi$. For $m,m'\in E_j$, taking $\pi$ to be the transposition of $m$ and $m'$ gives $q^\D_m=q^\D_{m'}$. So $q^\D$ is constant on each $E_j$, and the constraint $q(E_j)=c_j$ gives $q^\D_m=c_j/|E_j|$ for $m\in E_j$.
\end{proof}

\begin{proof}[Proof of Lemma~\ref{lem:admissible}]
Let $S=\supp q$. Here $\frac{\partial L_M(q)}{\partial q_m}$, $m\in S$, denote the partial derivatives of the restriction of $L_M$ to $\Delta S$ at $q$.

\emph{Only if.} Let $q=q^\D$ for some $\D=(M,\eta,\F)$. Since $\eta\in\Delta_\D$, Lemma~\ref{lem:properties}(a) gives $\eta\ac q$, and every element of $\Delta_\D$ is supported on $S$. Hence $q$ minimizes the restriction of $L_M$ to $\Delta S$ subject to the linear constraints $\sum_{m\in S}q_mf_m=\E_\eta [f]$, $f\in\F$, and $q$ lies in the relative interior of $\Delta S$. The objective function of the belief formation rule is a strictly convex loss with linear constraints, hence the first-order conditions establish existence of multipliers $\lambda_0$ and $(\lambda_f)_{f\in\F}$ such that
\[
\frac{\partial L_M(q)}{\partial q_m}=\lambda_0+\sum_{f\in\F}\lambda_f f_m\qquad\text{for all }m\in S.
\]
Since $q$ and $\eta$ are probability distributions supported on $S$ and $q\in\Delta_\D$, one can multiply both sides by $(q_m-\eta_m)$ for any $m \in S$, sum over $m \in S$, and obtain
\[
\varphi_L(q,\eta)=\lambda_0\sum_{m\in S}(q_m-\eta_m)+\sum_{f\in\F}\lambda_f\bigl(\E_q[f]-\E_\eta [f]\bigr)=0.
\]

\emph{If.} Let $\eta\ac q$ and $\varphi_L(q,\eta)=0$. Disclose the statistic $f$ with $f_m=\frac{\partial L_M(q)}{\partial q_m}$ for $m\in S$ and $f_m=0$ otherwise, together with the indicator $\one_{M\setminus S}$. Both $q$ and $\eta$ give probability zero to $M\setminus S$, and $\E_\eta[f]-\E_q[f]=\sum_{m\in S}\frac{\partial L_M(q)}{\partial q_m}(\eta_m-q_m)=0$. Hence $q,\eta\in\Delta_\D$, and every $q'\in\Delta_\D$ is supported on $S$ and satisfies $\E_{q'}[f]=\E_q[f]$. By convexity, one has
\[
L_M(q')\ \ge\ L_M(q)+\sum_{m\in S}\frac{\partial L_M(q)}{\partial q_m}(q'_m-q_m)=L_M(q)+\E_{q'}[f]-\E_q[f]=L_M(q)\qquad\text{for all }q'\in\Delta_\D.
\]
By strict convexity, $q$ is the unique minimizer, so $q=q^\D$.

The equivalence with the Bregman condition follows from the definition $D_L(\eta,q)=L_M(\eta)-L_M(q)+\varphi_L(q,\eta)$.
\end{proof}

\begin{proof}[Proof of Corollary~\ref{cor:shannon-admissible}]
For $L_M(q)=-H(q)=\sum_m q_m\log q_m$, one has $\frac{\partial L_M(q)}{\partial q_m}=\log q_m+1$ for $m\in S$. Since $\sum_{m\in S}(q_m-\eta_m)=0$, $\varphi_L(q,\eta)=\E_q[\log q]-\E_\eta[\log q]$, and Lemma~\ref{lem:admissible} applies. The equivalent forms follow from $H(\eta,q)=-\E_\eta[\log q]$ and $\KL(\eta, q)=\E_\eta[\log(\eta/q)]=H(\eta,q)-H(\eta)$.
\end{proof}

\subsection{Proofs for Section~\ref{sec:implementation}}

\begin{lemma}\label{lem:obedience}
Let $q\in\Delta M$ and $\sigma$ satisfy condition (ii) of Definition~\ref{def:impl}. Then $p=q\circ\sigma$ is a correlated equilibrium.
\end{lemma}

\begin{proof}
Fix $i$, $a_i$ and $a'_i\in A_i$. Then
\[
\begin{aligned}
&\sum_{m_i}\sigma_i(a_i\mid m_i)\sum_{m_{-i}}q(m_i,m_{-i})\bigl[u_i(a_i,\sigma_{-i}(m_{-i}))-u_i(a'_i,\sigma_{-i}(m_{-i}))\bigr]
\\
&\quad=\sum_{a_{-i}}p(a_i,a_{-i})\bigl[u_i(a_i,a_{-i})-u_i(a'_i,a_{-i})\bigr].
\end{aligned}
\]
By condition (ii), if $\sigma_i(a_i\mid m_i)>0$ and $q(m_i)>0$, then $a_i$ is in the support of a best response to $q(\cdot\mid m_i)$ and $\sigma_{-i}$, hence itself a best response.
\end{proof}

\begin{proof}[Proof of Theorem~\ref{thm:main}, $\text{(a) } \IMPL\subseteq\JC$]
Let $(\D,\sigma)$ implement $\mu$ and $q=q^\D$. By Lemma~\ref{lem:obedience}, $p=q\circ\sigma\in\CE$. If $\mu(a)>0$, there is $m$ with $\eta(m)\sigma(a\mid m)>0$. By Lemma~\ref{lem:properties}(a), $q(m)>0$, hence $p(a)\ge q(m)\sigma(a\mid m)>0$. Thus $\supp\mu\subseteq\supp p\subseteq\supp\CE$. 
\end{proof}

\begin{lemma}\label{lem:hypercube}
Let $d,k,r\in\N$ with $k\ge r\ge0$. There exists $C=(c_{i_1,\dots,i_d})_{1\le i_1,\dots,i_d\le k}\in\{0,1\}^{k^d}$ such that
\begin{equation}\label{eq:P}\tag{P}
\sum_{j=1}^{k}c_{i_1,\dots,i_{l-1},j,i_{l+1},\dots,i_d}=r
\end{equation}
for every $1\le l\le d$ and every $1\le i_1,\dots,i_{l-1},i_{l+1},\dots,i_d\le k$.
\end{lemma}

\begin{proof}
For $d=2$ this is a special case of Theorem~1.1, Chapter~6 of \citet{ryser1963combinatorial}. In general we proceed by induction on $d$. Indices are taken modulo $k$ in $\{1,\dots,k\}$ (so $0\equiv k$). For $d=1$, let $c_i=1$ if and only if $i\le r$.

Let $C$ satisfy \eqref{eq:P} in dimension $d$. For $s\in\{1,\dots,k\}$ let $C^s$ be the shift of $C$ in the first coordinate, $c^s_{i_1,\dots,i_d}=c_{i_1-s+1,i_2,\dots,i_d}$. Each $C^s$ satisfies \eqref{eq:P}: along the first coordinate, $\sum_jc^s_{j,i_2,\dots,i_d}=\sum_jc_{j-s+1,i_2,\dots,i_d}=r$, because $j\mapsto j-s+1$ is a bijection of $\{1,\dots,k\}$. Along any other coordinate $l\ge2$, the line of $C^s$ through $(i_1,\dots)$ is the line of $C$ through $(i_1-s+1,\dots)$.

Define $\tilde C\in\{0,1\}^{k^{d+1}}$ by $\tilde c_{i_1,\dots,i_d,i_{d+1}}=c^{i_{d+1}}_{i_1,\dots,i_d}=c_{i_1-i_{d+1}+1,i_2,\dots,i_d}$. For $l\le d$, a line of $\tilde C$ along coordinate $l$ is a line of $C^{i_{d+1}}$, so it sums to $r$. Along coordinate $d+1$,
\[
\sum_{j=1}^k\tilde c_{i_1,\dots,i_d,j}=\sum_{j=1}^kc_{i_1-j+1,i_2,\dots,i_d}=\sum_{j'=1}^kc_{j',i_2,\dots,i_d}=r,
\]
since $j\mapsto i_1-j+1$ is a bijection. Hence $\tilde C$ satisfies \eqref{eq:P} in dimension $d+1$.\footnote{The author thanks Gil Puig i Surroca for suggesting this approach.}
\end{proof}

\begin{proof}[Proof of Lemma~\ref{lem:construction}]
For $m\in M$ let $a(m)$ be the profile with $m\in M^{a(m)}$. For each $a$ with $\mu(a)>0$ we have $k_a\ge1$, so $C^a$ has a nonzero entry. Choose $m^a\in M^a$ with $C^a(m^a)=1$ and let $\eta=\sum_a\mu(a)\delta_{m^a}$. Then $\eta(Z)=0$ and $\eta\circ\sigma=\mu$.

\emph{Belief.} $\Delta_\D=\{q:q(Z)=0\}$. By Lemma~\ref{lem:properties}(b), $q^\D$ is uniform on $M\setminus Z$.
Fix a profile $a \in A$, a player $i$ and $m_i\in M_i^{a_i}$. The slice $\{m_i\}\times\times_{l\neq i}M_l^{a_l}$ of the block $M^a$ is an $(n-1)$-dimensional array that is the disjoint union of $k^{n-2}$ lines along the coordinate of some player $l\neq i$. By \eqref{eq:P} each line contains $k_a$ ones, so the slice contains $k^{n-2}k_a$ profiles outside $Z$. Summing over $m_i\in M_i^{a_i}$, $|M^a\setminus Z|=k^{n-1}k_a$, and summing over $a$, $|M\setminus Z|=k^{n-1}\sum_ak_a=k^n$. Hence $q^\D(M^a)=p_a$ and, for $m_i\in M_i^{a_i}$ with $p(a_i)>0$,
\[
q^\D\bigl(\{m:\sigma_{-i}(m_{-i})=a_{-i}\}\ \big|\ m_i\bigr)=\frac{k^{n-2}k_{(a_i,a_{-i})}}{\sum_{a_{-i}'}k^{n-2}k_{(a_i,a_{-i}')}}=p(a_{-i}\mid a_i).
\]
If $p(a_i)=0$ then $q^\D(m_i)=0$ and no condition is imposed.

\emph{Best response.} For $m_i\in M_i^{a_i}$ with $q^\D(m_i)>0$, the left-hand side of the best-response inequality (condition (ii)) equals $\frac1k\sum_{a_{-i}}p(a_i,a_{-i})[u_i(a_i,a_{-i})-u_i(a_i',a_{-i})]$, which is nonnegative since $p\in\CE$.
\end{proof}

\begin{proof}[Proof of Theorem~\ref{thm:main}]
(a),  $\JC \subseteq \bigcap_{\epsilon>0} \IMPL_\epsilon$. Let $\mu\in\JC$ and let $\bar p\in\CE$ have support $\supp\CE$ (e.g.\ the average of the vertices). Let $B=\max_{i,a,a_i'}|u_i(a)-u_i(a_i',a_{-i})|$. For $\delta>0$ choose a rational $p'\in\Delta A$ with $\supp p'=\supp\bar p$ and $\|p'-\bar p\|_1<\delta$. Then for all $i,a_i,a_i'$, $\sum_{a_{-i}}p'(a)[u_i(a)-u_i(a_i',a_{-i})]\ge-B\delta$. Consider the construction of Lemma~\ref{lem:construction} with $p'$ in place of $p$, writing $p'_a=k'_a/k$. The belief formation step does not use the equilibrium property, and for $m_i\in M_i^{a_i}$ one has $q(\{m_i\}\times\{m_{-i}:\sigma_{-i}(m_{-i})=a_{-i}\})=k^{n-2}k'_a/k^n=p'(a_i,a_{-i})/k$, so in the best-response step one has $\frac1k\sum_{a_{-i}}p'(a)[u_i(a)-u_i(a_i',a_{-i})]\ge-B\delta/k$. Choosing $\delta=\varepsilon/B$ gives $\varepsilon$-implementation of $\mu$.
\\
(b) If the game has rational payoffs, the set $\CE$ is a polytope defined by finitely many linear inequalities with rational coefficients, so its vertices $v^1,\dots,v^r$ have rational entries. Every $p\in\CE$ is a convex combination of vertices, so $\supp\CE=\bigcup_l\supp v^l=\supp\bar p$ with $\bar p=\frac1r\sum_lv^l$, which is a rational correlated equilibrium. We can apply Lemma~\ref{lem:construction} to $\bar p$, which yields the result.

\end{proof}

\begin{proof}[Proof of Corollary~\ref{cor:simple}]
The result follows the construction behind the proof of Theorem~\ref{thm:main}(b) and Lemma~\ref{lem:construction}.
\end{proof}

\subsection{Proofs for Section~\ref{sec:endogenous}}
\begin{proof}[Proof of Proposition~\ref{prop:endogenous}]
\emph{$\overline{\IMPL}\subseteq\IMPL$.} Let $\bar\D=(M,\bar\eta,\bar\F)$ implement $\mu$. Each player's data (i.e., a private message and an action) can be relabeled as a fictitious message: let $M_i'=M_i\times A_i$, with $m_i'=(m_i,a_i)$, so that $M'=\times_iM_i'$ coincides with $M\times A$. Let $\eta'=\bar\eta$, $\F'=\bar\F$, and the profile of (pure) strategies be defined from $\sigma_i'(m_i,a_i)=a_i$. Since the plausible sets, the selected belief is $\bar q$. Tthe best-response condition of Definition~\ref{def:impl} at $m_i'=(m_i,a_i)$ is the condition of Definition~\ref{def:endo} at $(m_i,a_i)$, and $\eta'\circ\sigma'=\marg_A\bar\eta=\mu$. Hence $\mu\in\IMPL$.

\emph{$\IMPL\subseteq\overline{\IMPL}$.} Let $\mu\in\IMPL$. By Corollary~\ref{cor:simple}, $\mu$ is implemented by some $(M,\eta,\{\one_Z\})$ and a pure profile $\sigma$. Consider a feedback structure that discloses the strategy profile as follows: let $G=\{(m,\sigma(m)):m\in M\}$, $\bar Z=\bigl((M\times A)\setminus G\bigr)\cup\{(m,\sigma(m)):m\in Z\}$, $\bar\F=\{\one_{\bar Z}\}$. The set $\bar Z$ has two parts. The first, $(M\times A)\setminus G$  discloses the strategy profile, as it consists of every $(m,a)$ with $a\neq\sigma(m)$, for every $m\in M$. The second, $\{(m,\sigma(m)):m\in Z\}$, consists of the pairs at the excluded messages, and transfers the zero set $Z$ of the original implementation. Since $\bar q(\bar Z)=0$ if and only if $\bar q$ vanishes on every pair in $\bar Z$, the single statistic $\one_{\bar Z}$ conveys both. Let $\bar\eta(m,a)=\eta(m)\one[a=\sigma(m)]$, so that $\marg_M\bar\eta=\eta$, $\marg_A\bar\eta=\eta\circ\sigma=\mu$, and $\bar\eta(\bar Z)=0$. By Lemma~\ref{lem:properties}(b), $q^\D$ is uniform on $M\setminus Z$ and $\bar q$ is uniform on $\{(m,\sigma(m)):m\in M\setminus Z\}$, so $\bar q(m,a)=q^\D(m)\one[a=\sigma(m)]$. If $\bar q(m_i,a_i)>0$, then $a_i=\sigma_i(m_i)$ and $q^\D(m_i)>0$, and the best-response condition of Definition~\ref{def:endo} at $(m_i,a_i)$ coincides with condition (ii) of Definition~\ref{def:impl} at $m_i$. Hence $\bar\D$ implements $\mu$.

\end{proof}

\subsection{Proofs for Section~\ref{sec:direct}}

\begin{proof}[Proof of Proposition~\ref{prop:direct}]
With $M=A$ and obedient $\sigma$, condition (ii) of Definition~\ref{def:impl} at $a_i$ with $q(a_i)>0$ is the inequality that defines a correlated equilibrium, so it holds if and only if $q\in\CE$. Condition (iii) implies that $\mu=\eta$. Hence $\mu$ is directly implementable if and only if there exist $\F$ and $q\in\CE$ with $q=q^{(A,\mu,\F)}$, that is, the couple $(q,\mu)$ is admissible given loss $L$. Lemma~\ref{lem:admissible} yields the equivalence of (i) and (ii), and Corollary~\ref{cor:differentiable} the differentiable case.
\end{proof}

\begin{proof}[Proof of Corollary~\ref{cor:direct-shannon}]
The first statement combines Proposition~\ref{prop:direct} and Corollary~\ref{cor:shannon-admissible}. With a unique correlated equilibrium the union in Proposition~\ref{prop:direct}(iii) has one element: it is the intersection of the face $\Delta(\supp q^*)$ with an affine hyperplane (hence a convex polytope).
\end{proof}

\subsection{Proofs for Section~\ref{sec:decision}}\label{app:decision}

\begin{proof}[Proof of Proposition~\ref{prop:decision}]
(a) Let $\mu(a)>0$. There is $m$ with $\sum_\theta\eta(\theta,m)>0$ and $\sigma(a\mid m)>0$. By Lemma~\ref{lem:properties}(a), $q(\theta,m)>0$ whenever $\eta(\theta,m)>0$, so $q(m)>0$ and $a\in\BR(q(\cdot\mid m))$.

(b) Let $\supp\mu=\{a_1,\dots,a_K\}\subseteq R^\circ$ and write $\mu_b=\mu(a_b)$.

\emph{Step 1: rational interior beliefs.} For each $b$, the set $P_b=\{p\in\Delta\Theta:a_b\in\BR(p)\}$ is a polytope defined by linear inequalities with rational coefficients and contains a full-support belief $p$. Let $r\in\relint P_b$. Points $(1-s)p+sr$, $s\in(0,1]$, lie in $\relint P_b$ and, for $s$ small, have also full support. Since rational points are dense in $\relint P_b$, there is a rational $p^*_b\in P_b$ with full support.

\emph{Step 2: messages and feedback.} For each action $b \in R^\circ$ and $\theta \in \Theta$, let $n_{b,\theta}=C_b\,p_0(\theta)/p^*_b(\theta)$, where $C_b\in\N$ is chosen so that all $n_{b,\theta}$ are positive integers. Let $R_b=\max_\theta n_{b,\theta}$, $M_b=\{(b,1),\dots,(b,R_b)\}$ and $M=\bigcup_bM_b$. Let $Z=\{(\theta,(b,r)):r>n_{b,\theta}\}$ and $E_{b,\theta}=\{\theta\}\times M_b$.

\emph{Step 3: true process and feedback.} Let $\eta(\theta,(b,1))=p_0(\theta)\mu_b$ and $\eta=0$ elsewhere. Then $\mathrm{marg}_\Theta\eta=p_0$, $\eta(Z)=0$ (since $n_{b,\theta}\ge1$), and $\eta(E_{b,\theta})=p_0(\theta)\mu_b$. Let $\F=\{\one_Z\}\cup\{\one_{E_{b,\theta}}:b,\theta\}$. 

\emph{Step 4: belief.} The uncertainty set is $\Delta_\D=\{q:q(Z)=0,\ q(E_{b,\theta}\setminus Z)=p_0(\theta)\mu_b\ \forall b,\theta\}$, and the sets $Z$ and $E_{b,\theta}\setminus Z$ partition $\Theta\times M$. Since $\sum_b\one_{E_{b,\theta}}=\one_{\{\theta\}\times M}$, every $q\in\Delta_\D$ has $\mathrm{marg}_\Theta q=p_0$ (i.e., the prior is disclosed). By Lemma~\ref{lem:properties}(b), $q(\theta,(b,r))=p_0(\theta)\mu_b/n_{b,\theta}=\mu_bp^*_b(\theta)/C_b$ for $r\le n_{b,\theta}$. In particular, at message $(b,1)$, which has no zero entries, $q(\theta\mid(b,1))=p^*_b(\theta)$.

\emph{Step 5: strategy and outcome.} Let $\sigma(b,1)=a_b$, that is optimal by Step 1, and let $\sigma(b,r)$ for $r\ge2$ be any best response to $q(\cdot\mid(b,r))$. Since $\eta$ is concentrated on the messages $(b,1)$, the outcome assigns probability $\sum_\theta\eta(\theta,(b,1))=\mu_b$ to action $a_b$.

(c) follows from (a) and (b).
\end{proof}

\paragraph{Use of AI:} During the preparation of this work the author used Claude Opus 5.5 and Chat GPT 5.6 improve the writing, formatting and proof-checking. After using this tool/service, the author reviewed and edited the content as needed and take full responsibility for the content of the published article.


\end{document}